\documentclass[11pt,a4paper]{article}

\usepackage[a4paper,margin=1in]{geometry}
\usepackage[T1]{fontenc}
\usepackage[utf8]{inputenc}
\usepackage{lmodern}
\usepackage{microtype}
\usepackage{amsmath,amssymb,amsthm,mathtools}
\usepackage{bm}
\usepackage{mathrsfs}
\usepackage{graphicx}
\usepackage{booktabs}
\usepackage{array}
\usepackage{xcolor}
\usepackage{placeins}
\usepackage{needspace}
\usepackage[
  colorlinks=true,
  linkcolor=blue,
  citecolor=blue,
  urlcolor=blue
]{hyperref}

\allowdisplaybreaks[2]

\newtheorem{theorem}{Theorem}[section]
\newtheorem{proposition}[theorem]{Proposition}

\newtheorem{lemma}[theorem]{Lemma}

\theoremstyle{definition}

\newcommand{\Mfp}{\mathcal{M}}
\newcommand{\Aff}{\mathfrak{A}}
\newcommand{\Gau}{\mathfrak{G}}
\newcommand{\Hor}{\mathcal{H}}
\newcommand{\Hilb}{\mathscr{H}}
\newcommand{\Qdom}{\mathcal{Q}}
\newcommand{\ad}{\operatorname{ad}}
\newcommand{\Ran}{\operatorname{Ran}}
\newcommand{\Ker}{\operatorname{Ker}}
\newcommand{\spec}{\operatorname{spec}}

\newcommand{\sgn}{\operatorname{sgn}}

\newcommand{\one}{\mathbb{1}}
\newcommand{\dd}{\mathrm{d}}
\newcommand{\iu}{\mathrm{i}}
\newcommand{\norm}[1]{\left\lVert#1\right\rVert}
\newcommand{\abs}[1]{\left\lvert#1\right\rvert}
\newcommand{\inner}[2]{\left\langle#1,#2\right\rangle}
\newcommand{\ket}[1]{\lvert#1\rangle}

\newcommand{\dyad}[2]{\lvert#1\rangle\langle#2\rvert}

\begin{document}

\title{Some Remarks on the Spectral Geometry of the Gribov Horizon}

\author{
Daniel G. Tedesco\\
\small ESEHL/PPGENT, UNINTER, Curitiba, Paran\'a, Brazil\\
\small \texttt{daniel.te@uninter.com}
}

\date{}
\maketitle

\begin{abstract}
We develop a local spectral framework for the Landau-gauge Gribov horizon that
distinguishes gauge-orbit projection from degeneracy of the gauge-fixing
Hessian. On a flat torus, spatially constant ghosts form a residual global-color
kernel; after removal of this kernel, the reduced Faddeev-Popov operator is the
normal Morse-Bott Hessian of the orbit-norm functional, whereas the covariant
Laplacian defines the orthogonal connection. For the associated affine focal
pencil, we prove a quadratic-form index theorem with a Morse-Bott endpoint. At
a regular isolated crossing, the critical projector $P$ and invertible
compressed derivative $\Gamma=P\dot{\Mfp}P$ determine the spectral-flow jump and
leading Laurent coefficient of the sourced ghost resolvent; for a simple zero,
they also give the wall conormal. Crossings reached
along affine rays from the positive region have negative-definite $\Gamma$,
including symmetry-protected multiplets, while a second Schur reduction
determines pole orders along tangential paths. A projected single-harmonic model
checks the projected Feynman-Hellmann relation. For an $SU(2)$ hedgehog on
$\mathbb R^3$,
we construct a normalizable threshold state in every coupled spin-orbit channel
and minimize over the full tower to obtain the exact stability interval
$-2<g<1$. Dirichlet-box spectra approach these thresholds and serve as
finite-volume comparisons; no numerical fit enters the continuum result.
\end{abstract}


\section{Introduction}
\label{sec:introduction}

Gauge fixing provides local coordinates on the configuration space of a
non-Abelian gauge theory, but it does not furnish a global section of the gauge
bundle. Gribov first showed that the Landau condition may intersect a single
gauge orbit more than once, thereby invalidating its interpretation as a global
coordinate condition~\cite{Gribov1978}. Singer subsequently identified the
topological obstruction to a continuous global gauge section on the
irreducible configuration space~\cite{Singer1978}. Geometric approaches related
these observations to the Riemannian geometry of orbit space, tthe restriction of the norm functional to a gauge orbit, focal points, and gauge-invariant coordinate descriptions
of the orbit-space quotient ~\cite{NarasimhanRamadas1979,BabelonViallet1981,
OrlandMetric1997,OrlandSemenoff2000,OrlandGaugeCoordinates2004}. Within this framework, the first Gribov region is bounded by configurations for which the
lowest nontrivial eigenvalue of the Landau Faddeev-Popov operator vanishes
~\cite{MaskawaNakajima1978,MaskawaNakajima1980,SemenovFranke1982,DellAntonioZwanziger1991,vanBaal1992}.

The quantum implementation of the restriction to the first Gribov region is
addressed by the Gribov-Zwanziger program. Its horizon condition, localizable
action, relation to the no-pole condition, and later BRST-invariant formulations
concern the functional measure and its gauge dependence
~\cite{Zwanziger1989,VandersickelZwanziger2012,CapriNoPole2013,CapriBRST2016,CapriUniversal2018}.
We examine the local spectral geometry of a fixed smooth Landau-gauge
background. This level of analysis determines the singular response of the
ghost operator and the orientation of a prescribed path through the first
horizon; it does not determine the global modular region or the statistical
weight assigned to configurations near its boundary. Lattice studies connect
the lowest nontrivial Faddeev-Popov eigenvalues with the ghost sector and the
approach to the first horizon
~\cite{Greensite2010,CucchieriMendesGhost2013,CucchieriMendesCrossing2014},
while Henyey-type constructions provide analytic backgrounds on which the
corresponding local statements can be tested
~\cite{Henyey1979,GuimaraesSorella2011,CapriGuimaraesSorellaTedesco2012,Capri:2013vka}.
Earlier one-mode and Morse-theoretic analyses likewise describe nontrivial local
geometry before a functional-integral interpretation is imposed
~\cite{Carson1986,Landim2014}. The aim here is to place these ingredients in an
operator framework applicable to both analytic and finite-volume calculations.

The construction begins with a distinction that is sometimes obscured because
two elliptic operators act on the same adjoint scalar fields. The orthogonal
decomposition of a tangent vector into orbit-tangent and orbit-normal parts is
controlled by the covariant Laplacian
\begin{equation}
  \Delta_A=D_A^\ast D_A.
  \label{eq:intro-delta}
\end{equation}
By contrast, degeneracy of the Landau condition is controlled by the
Faddeev-Popov operator
\begin{equation}
  \Mfp[A]=-\partial_\mu D_\mu[A].
  \label{eq:intro-fp}
\end{equation}
On a transverse background, these operators are related by
\begin{equation}
  \Mfp[A]=\Delta_A+\ad_{A_\mu}D_\mu.
  \label{eq:intro-two-operators}
\end{equation}
After the Hessian of the orbit norm is pulled back to ghost parameters, the
first term in Eq.~\eqref{eq:intro-two-operators} is the orbit-metric
contribution, whereas the second is the second fundamental form paired with the
normal vector $A$. Consequently, a nontrivial Faddeev-Popov zero mode at an
irreducible connection need not belong to $\Ker\Delta_A$. The orthogonal
connection on gauge-field space may therefore remain regular at a Landau-gauge
horizon even though the ghost inverse is singular there.

On a torus, the Faddeev-Popov operator also possesses a residual kernel that is
present at every Landau configuration: every spatially constant adjoint scalar
is annihilated by $\Mfp[A]$. These directions generate global color rotations
of the norm functional, whose critical set is therefore Morse-Bott rather than
Morse. We remove this fixed symmetry kernel before defining the first Gribov
region. Unless stated otherwise, all eigenvalues, kernels, indices, and
determinants below refer to the reduced ghost space. A point lies on the first
nontrivial horizon only when the reduced operator acquires an additional zero
mode.

The local spectral data of such a crossing consist of the critical projector
$P$ and the compressed first variation
\begin{equation}
  \Gamma_B=P\,\delta_B\Mfp[A_\ast]\,P.
  \label{eq:intro-crossing-matrix}
\end{equation}
For a normalized simple zero mode $f_0$, this operator-valued crossing form
reduces to the scalar
\begin{equation}
  \Gamma_A(B;f_0)
  =\inner{f_0}{\delta_B\Mfp[A]f_0}
  =\int_M(\partial_\mu f_0^a)[B_\mu,f_0]^a\,\dd^d x.
  \label{eq:intro-crossing-scalar}
\end{equation}
Together, $P$ and $\Gamma_B$ determine the conormal to a simple spectral wall,
the orientation of the corresponding spectral-flow event, and the leading
Laurent coefficient of a sourced inverse. The projector is indispensable in
the last statement because the residue depends on the component of the source
inside the critical eigenspace.

Affinity yields a direct operator bound along affine rays. Suppose that
$\Mfp_t=\Mfp_0+tV$ and that $\Mfp_0\ge\lambda_0(0)>0$ on the reduced ghost
space. If $P$ projects onto $\Ker\Mfp_{t_\ast}$ at some $t_\ast>0$, affinity
gives the exact compression identity and bound
\begin{equation}
  PVP=-\frac{1}{t_\ast}P\Mfp_0P
  \le -\frac{\lambda_0(0)}{t_\ast}P.
  \label{eq:intro-affine-bound}
\end{equation}
Thus every zero crossing reached along a ray from a positive configuration is
regular in the radial direction, including crossings with a degenerate
critical eigenspace. The same affine dependence makes the lowest eigenvalue a
concave function along line segments and recovers the convexity of the first
Gribov region. This statement applies to the specified outward ray; it does not
exclude directions tangent to the spectral wall in the full configuration
space.

To determine the associated ghost singularity, we use the
Feshbach-Schur map~\cite{GriesemerHasler2008}. At an isolated crossing, the
coupling between the critical eigenspace and its spectral complement vanishes
at threshold, so the self-energy first contributes at second order in the path
parameter. This fact accounts for the appearance of $\Gamma_B^{-1}$ in the
coefficient of the simple pole at a regular crossing. If the crossing form has
a nontrivial kernel along a general analytic path, however, the regular part of
the critical space must be eliminated by a second Schur complement. The first
nonzero coefficient of the resulting twice-reduced operator then fixes the pole
order and its leading coefficient. A Moore-Penrose pseudoinverse of
$\Gamma_B$ does not, in general, encode this higher-order reduction.

Two complementary models delimit the reach of the analysis. A
single-harmonic, Cartan-valued background on $T^2$ produces a reducible,
sector-projected Mathieu-type matrix. It serves as a controlled test of the
Feynman-Hellmann identity and projected spectral convergence, but it is not
interpreted as an irreducible Gribov horizon. An $SU(2)$ hedgehog on
$\mathbb R^3$ instead
reduces the Faddeev-Popov equation to radial spin-orbit channels. Its zero
modes sit at the threshold of the essential spectrum, so the continuum problem
is an analogue of the compact-operator results rather than a direct instance of
them; the corresponding Dirichlet-box realization supplies the finite-volume
test. For one Henyey family, a normalizable threshold state can be constructed
in every coupled channel, and minimization over the complete channel tower
places the first thresholds at $g=1$ and $g=-2$. The resulting exact interval
$-2<g<1$ contains no negative radial eigenvalue. Finally, the protected $J=1$
triplet is used to display symmetry-allowed splitting matrices in
representation space. These matrices are illustrative: no transverse
gauge-field deformation that realizes them is claimed.

The principal results established here are the form-domain focal index theorem
with explicit Morse-Bott endpoint splitting, the twice-reduced resolvent
analysis of tangential crossings, and the all-channel minimization of the
hedgehog thresholds. The Morse interpretation of the gauge-fixing Hessian,
local Faddeev-Popov spectral analysis, Feshbach reduction, and Henyey zero modes provide the established framework in which these results are proved
~\cite{BabelonViallet1981,OrlandSemenoff2000,Carson1986,
CucchieriMendesCrossing2014,Landim2014,
CapriGuimaraesSorellaTedesco2012,Capri:2013vka}.

\section{Functional setting and gauge-orbit geometry}
\label{sec:setting}

Let $M=T_L^d$ be the flat Euclidean torus of side length $L$, take the gauge
group to be $SU(N)$, and work on the globally trivial adjoint bundle. We use
smooth fields for the geometric arguments and pass to Sobolev completions for
the operator statements. The second-order elliptic operators have domain
$H^2(M,\mathfrak{su}(N))$ in the Hilbert space
\begin{equation}
  \Hilb=L^2(M,\mathfrak{su}(N)),
  \qquad
  \inner{f}{g}=\int_M\operatorname{tr}(f^\dagger g)\,\dd^d x,
  \label{eq:l2-space}
\end{equation}
and their closed quadratic forms have common domain $H^1$. Periodic boundary
conditions remove all boundary contributions from the integrations by parts
used below.

We adopt the conventions
\begin{equation}
  D_\mu[A]=\partial_\mu+\ad_{A_\mu},
  \qquad
  D_A^\ast=-D_\mu,
  \qquad
  \partial_\mu A_\mu=0.
  \label{eq:conventions}
\end{equation}
The last equality is the Landau condition. If
$c\in\mathfrak{su}(N)$ is constant on $M$, transversality implies
\begin{equation}
  \Mfp[A]c
  =-\partial_\mu[A_\mu,c]
  =-[\partial_\mu A_\mu,c]
  =0.
  \label{eq:constant-zero-modes}
\end{equation}
Denote the finite-dimensional space of these constant sections by
$\mathfrak g_{\mathrm{const}}$, and let $\Pi_0$ be the orthogonal projector onto
its complement. The reduced ghost Hilbert space and reduced Faddeev-Popov
operator are
\begin{equation}
  \Hilb_0=\Hilb\ominus\mathfrak g_{\mathrm{const}},
  \qquad
  \widehat{\Mfp}[A]=
  \Pi_0\Mfp[A]\Pi_0\big|_{\Hilb_0}.
  \label{eq:reduced-ghost-space}
\end{equation}
Because $\Mfp[A]$ is self-adjoint and annihilates
$\mathfrak g_{\mathrm{const}}$, the orthogonal complement $\Hilb_0$ is
invariant. Hence $\widehat{\Mfp}[A]$ is the restriction of $\Mfp[A]$ to
$\Hilb_0$; the projectors in Eq.~\eqref{eq:reduced-ghost-space} keep the symmetry
reduction explicit. Writing $\lambda_0(A)$ for the lowest eigenvalue of the
reduced operator, we define
\begin{equation}
  \Omega=\{A:\partial\!\cdot A=0,\ \lambda_0(A)>0\}.
  \label{eq:first-region}
\end{equation}
The boundary of $\Omega$ is the first nontrivial Gribov horizon. The terminology
``first Gribov region'' is also commonly applied to its closure, obtained by
replacing the strict inequality with $\lambda_0\ge0$.

The kernel of the covariant Laplacian is governed by a different condition.
Indeed,
\begin{equation}
  \inner{\xi}{\Delta_A\xi}
  =\inner{D_A\xi}{D_A\xi}\ge0,
  \qquad
  \Ker\Delta_A=\Ker D_A.
  \label{eq:delta-positive}
\end{equation}
A connection is irreducible when it has no covariantly constant adjoint scalar
after removal of the infinitesimal center. For $SU(N)$, $\Delta_A$ is therefore
strictly positive on the irreducible stratum. The constant ghosts in
Eq.~\eqref{eq:constant-zero-modes} are generally not covariantly constant, so
their presence in $\Ker\Mfp[A]$ is consistent with the strict positivity of
$\Delta_A$.

\subsection{The orthogonal connection and the Landau operator}
\label{subsec:two-operators}

Let $\Aff$ denote the affine space of connections and $\Gau$ the gauge group.
At an irreducible connection, the $L^2$ metric gives the orthogonal decomposition
\begin{equation}
  T_A\Aff=\Ran D_A\oplus\Ker D_A^\ast.
  \label{eq:orthogonal-splitting}
\end{equation}
For $a\in T_A\Aff$, write $a=a_H+D_A\xi$ with
$a_H\in\Ker D_A^\ast$. Applying $D_A^\ast$ and using the invertibility of
$\Delta_A=D_A^\ast D_A$ at an irreducible connection determines the vertical
parameter:
\begin{equation}
  \xi=(D_A^\ast D_A)^{-1}D_A^\ast a.
  \label{eq:vertical-parameter}
\end{equation}
The corresponding orthogonal connection form and horizontal projector are
therefore
\begin{equation}
  \omega_A(a)=\Delta_A^{-1}D_A^\ast a,
  \qquad
  \Pi_A^H a=a-D_A\Delta_A^{-1}D_A^\ast a.
  \label{eq:singer-connection}
\end{equation}
This local representative of Singer's orthogonal connection depends on the
inverse covariant Laplacian. Its local definition must be distinguished from
Singer's global topological obstruction, which is not caused by a local failure
of $\Delta_A^{-1}$ on the irreducible stratum.

Expanding $D_\mu D_\mu$ and comparing the result with
$-\partial_\mu D_\mu$ gives
\begin{align}
  D_A^\ast D_A\xi
  &=-\partial_\mu D_\mu\xi-[A_\mu,D_\mu\xi],
  \label{eq:two-operator-identity-a}\\
  \Mfp[A]
  &=\Delta_A+\ad_{A_\mu}D_\mu.
  \label{eq:two-operator-identity-b}
\end{align}
The two operators have the same principal symbol, $-\partial^2$, but differ in
their lower-order terms. Accordingly, if $\Mfp[A]\xi=0$ at an irreducible
background, then
\begin{equation}
  \Delta_A\xi=-\ad_{A_\mu}D_\mu\xi,
  \label{eq:delta-on-fp-mode}
\end{equation}
and the right-hand side need not vanish. A nontrivial Landau-gauge horizon can
therefore be encountered while the orthogonal connection remains regular.

For horizontal constant extensions
$\tau_1,\tau_2\in\Ker D_A^\ast$, the curvature of
Eq.~\eqref{eq:singer-connection} is
\begin{equation}
  \mathcal R_A(\tau_1,\tau_2)
  =-2\Delta_A^{-1}[\tau_{1\mu},\tau_{2\mu}].
  \label{eq:singer-curvature}
\end{equation}
Equation~\eqref{eq:singer-curvature} identifies the inverse covariant Laplacian
that enters the orbit-space geometry
~\cite{BabelonViallet1981,OrlandMetric1997,
OrlandGaugeCoordinates2004,ONeill1966,MoncriefMariniMaitra2019}. Its regularity
is fixed by irreducibility, whereas the ghost singularity analyzed below is
controlled by $\widehat{\Mfp}[A]^{-1}$.

\subsection{The norm functional and its normal Hessian}
\label{subsec:norm}

Consider the squared norm restricted to the gauge orbit of $A$,
\begin{equation}
  V_A[U]=\frac12\norm{A^U}^2.
  \label{eq:norm-functional}
\end{equation}
For the one-parameter family $U(t)=e^{t\xi}$, the tangent to the orbit at
$t=0$ is $D_A\xi$. The first variation is
\begin{equation}
  \frac{\dd}{\dd t}V_A[e^{t\xi}]\Big|_{t=0}
  =\inner{A}{D_A\xi}
  =-\inner{\partial_\mu A_\mu}{\xi},
  \label{eq:norm-first-variation}
\end{equation}
so the Landau condition characterizes the critical points of the orbit norm.
Differentiating once more yields
\begin{align}
  \frac{\dd^2}{\dd t^2}V_A[e^{t\xi}]\Big|_{t=0}
  &=\norm{D_A\xi}^2+
    \int_M A_\mu^a[(D_A\xi)_\mu,\xi]^a\,\dd^d x
  \notag\\
  &=\inner{\xi}{\bigl(\Delta_A+\ad_{A_\mu}D_\mu\bigr)\xi}
   =\inner{\xi}{\Mfp[A]\xi}.
  \label{eq:norm-second-variation}
\end{align}
This identification of the Landau Faddeev-Popov operator with the
second variation of the orbit-norm functional, together with its role
in testing local minimality along a gauge orbit, was used explicitly
in the orbit-space analysis of Ref.~\cite{OrlandSemenoff2000}. Global color rotations leave $V_A$ unchanged and generate the fixed kernel in
Eq.~\eqref{eq:constant-zero-modes}. The normal Hessian of the resulting
Morse-Bott critical manifold is consequently $\widehat{\Mfp}[A]$ acting on
$\Hilb_0$.

After pullback from orbit tangents to ghost parameters by $D_A$, the
Hessian of the orbit-norm functional decomposes as
\begin{equation}
  \operatorname{Hess}V_A(D_A\xi,D_A\eta)
  =\underbrace{\inner{\eta}{\Delta_A\xi}}_{\text{orbit metric}}
   +\underbrace{\inner{\eta}{\ad_{A_\mu}D_\mu\xi}}_
     {\text{second fundamental form paired with }A}.
  \label{eq:gauss-decomposition}
\end{equation}
This geometric interpretation applies to the pulled-back bilinear form. The
operator difference in Eq.~\eqref{eq:two-operator-identity-b} should not be
identified, without qualification, with a shape operator on all of
$T_A\Aff$. Only when the chosen representative realizes the absolute minimum of
the norm along its gauge orbit can $V_A$ be identified with the squared
orbit-space distance from the pure-gauge orbit
~\cite{OrlandMetric1997,OrlandSemenoff2000}.
Related hypersurface descriptions of gauge-orbit space, including
non-Abelian Hodge-type coordinates, intrinsic and extrinsic curvature,
and coordinate degeneracies associated with the Gribov boundary, were
developed in ~\cite{OrlandGaugeCoordinates2004}.

\subsection{The focal pencil and its Morse-Bott endpoint}
\label{subsec:focal-pencil}

Fix the affine slice through $A$,
\begin{equation}
  S_A=\{A+\tau:D_A^\ast\tau=0\},
  \label{eq:fixed-slice}
\end{equation}
and consider the normal segment $A_s=(1-s)A$, $0\le s\le1$, directed from
$A$ toward the origin. The orbit through $A_s$ is tangent to this fixed slice
when an infinitesimal gauge parameter $\xi$ satisfies
$D_A^\ast D_{A_s}\xi=0$. The operator pencil encoding this tangency condition
is
\begin{equation}
  T_A(s)=D_A^\ast D_{(1-s)A}.
  \label{eq:full-focal-pencil}
\end{equation}
Using transversality, one obtains
\begin{equation}
  D_A^\ast\ad_A+\ad_{A_\mu}D_\mu=0,
  \label{eq:ad-identity}
\end{equation}
and hence
\begin{equation}
  T_A(s)=\Delta_A+sC_A,
  \qquad
  C_A=\ad_{A_\mu}D_\mu,
  \qquad
  T_A(0)=\Delta_A,\quad T_A(1)=\Mfp[A].
  \label{eq:focal-pencil}
\end{equation}
For irreducible $A$, $\Delta_A$ is strictly positive on $\Hilb$, since a zero
mode would be a covariantly constant adjoint section and thus an infinitesimal
stabilizer of the connection. We equip its form domain
$\Qdom(\Delta_A)=H^1$ with the energy norm
$\norm{\xi}_{A}=\norm{\Delta_A^{1/2}\xi}$.

The geometric count requires the full pencil in Eq.~\eqref{eq:focal-pencil}.
For a generic background, neither $\Delta_A$ nor $C_A$ preserves $\Hilb_0$;
compressing $T_A(s)$ by $\Pi_0$ along the open segment could therefore create
singular parameters that do not represent tangencies between the gauge orbit
and the fixed slice. Reduction by global color is legitimate at the endpoint
$s=1$, where the self-adjoint Faddeev-Popov operator supplies the required
invariant orthogonal decomposition.

\begin{theorem}[Focal index theorem with a Morse-Bott endpoint]
\label{thm:focal-index}
Let $A$ be a smooth, transverse, irreducible connection on $T_L^d$. Define
\begin{equation}
  K_A=\Delta_A^{-1/2}C_A\Delta_A^{-1/2}
  \quad\text{on }\Hilb.
  \label{eq:KA}
\end{equation}
Then $K_A$ is compact and self-adjoint. The full pencil counts the reduced Morse
index according to
\begin{equation}
  n_-\bigl(\widehat{\Mfp}[A]\bigr)
  =n_-\bigl(\Mfp[A]\bigr)
  =\sum_{0<s<1}\dim\Ker T_A(s).
  \label{eq:focal-index}
\end{equation}
The sum counts focal parameters with multiplicity. At the endpoint,
\begin{equation}
  \Ker T_A(1)
  =\mathfrak g_{\mathrm{const}}\oplus
    \Ker\widehat{\Mfp}[A],
  \qquad
  \dim\Ker T_A(1)-\dim\mathfrak g_{\mathrm{const}}
  =\dim\Ker\widehat{\Mfp}[A].
  \label{eq:focal-endpoint-splitting}
\end{equation}
\end{theorem}

\begin{proof}
The quadratic form associated with $C_A$ is symmetric because both
$\Mfp[A]$ and $\Delta_A$ are symmetric on a transverse background. Since
$C_A$ has differential order one, the sandwiched operator in
Eq.~\eqref{eq:KA} has order $-1$ and is compact on the compact manifold. It is
also self-adjoint. Consider the map
\begin{equation}
  U_A:\Qdom(\Delta_A)\longrightarrow\Hilb,
  \qquad
  U_A\xi=\Delta_A^{1/2}\xi.
  \label{eq:energy-isometry}
\end{equation}
When $\Qdom(\Delta_A)$ carries the energy norm, $U_A$ is an isometric
isomorphism with inverse $\Delta_A^{-1/2}:\Hilb\to\Qdom(\Delta_A)$. This is an
isomorphism between the energy space and $\Hilb$; it does not treat the
unbounded operator $\Delta_A^{1/2}$ as a bounded operator on all of $\Hilb$.

Transporting the closed quadratic form of the pencil through the isometry in
Eq.~\eqref{eq:energy-isometry} gives
\begin{equation}
  t_s[\xi]
  =\inner{U_A\xi}{(\one+sK_A)U_A\xi}.
  \label{eq:form-equivalence}
\end{equation}
The form-domain isomorphism preserves both kernels and the maximal dimension of
negative subspaces. Let $\kappa_j$ denote the eigenvalues of $K_A$, counted
with multiplicity. A focal parameter $s\in(0,1)$ occurs at
$s=-1/\kappa_j$ precisely when $\kappa_j<-1$. Applying the min-max principle
to Eq.~\eqref{eq:form-equivalence} gives
\begin{equation}
  n_-\bigl(T_A(1)\bigr)
  =\#\{j:\kappa_j<-1\},
  \label{eq:minmax-index}
\end{equation}
again with multiplicity. The right-hand side is exactly the number of singular
parameters on the open segment.

At $s=1$, every constant adjoint section belongs to $\Ker\Mfp[A]$. Since
$\Mfp[A]$ is self-adjoint, $\Hilb_0$ is invariant and
$\Mfp[A]=0\oplus\widehat{\Mfp}[A]$ with respect to
$\mathfrak g_{\mathrm{const}}\oplus\Hilb_0$. The full and reduced operators
therefore have identical negative spectra, while their kernels differ by the
fixed global-color summand displayed in
Eq.~\eqref{eq:focal-endpoint-splitting}. These observations prove both
statements.
\end{proof}

The theorem is a quadratic-form result for the full geometric pencil; it does
not rely on treating an infinite-dimensional operator family as a finite
matrix. It also records the Morse-Bott character of the endpoint. If
$A\in\Omega$, the open segment contains no focal parameter. At the first
nontrivial horizon, $T_A(1)$ acquires an endpoint focal direction in addition
to the fixed global-color kernel. The residual multiplicity is removed only
after this endpoint decomposition has been made.

\subsection{Slice-adapted extraction and the geometric source}
\label{subsec:oblique}

Away from the horizon, a displacement can be decomposed relative to the fixed
Landau slice as
\begin{equation}
  a=a_T+D_A\xi,
  \qquad
  \partial_\mu a_{T\mu}=0.
  \label{eq:oblique-decomposition}
\end{equation}
Taking the ordinary divergence, rather than the covariant divergence used in
the orthogonal decomposition, gives
\begin{equation}
  \partial_\mu a_\mu=-\Mfp[A]\xi,
  \qquad
  \xi=-\widehat{\Mfp}[A]^{-1}\partial_\mu a_\mu,
  \label{eq:oblique-extraction}
\end{equation}
where $\partial\!\cdot a$ has zero spatial average and hence belongs to
$\Hilb_0$. Choosing the reduced space also fixes the otherwise undetermined
constant component of the gauge parameter. The map
\begin{equation}
  \omega_A^L(a)
  =-\widehat{\Mfp}[A]^{-1}\partial_\mu a_\mu
  \label{eq:landau-extraction}
\end{equation}
therefore extracts the vertical parameter relative to the fixed linear slice.
This construction defines a local oblique splitting. We do not identify it
with a gauge-equivariant principal connection because $\Ker\partial$ has not
been shown to define such a connection on the full gauge bundle.

Equation~\eqref{eq:landau-extraction} also identifies the geometric source of
the ghost pole. Let $A_t$ cross a simple horizon at $t=t_\ast$, and let $f_0$
be the normalized critical mode. The resolvent analysis in
Section~\ref{sec:resolvent} yields
\begin{equation}
  \omega_{A_t}^L(a)
  =-\frac{\ket{f_0}\inner{f_0}{\partial_\mu a_\mu}}
          {\Gamma(t-t_\ast)}+O(1).
  \label{eq:oblique-pole-preview}
\end{equation}
The residue depends on the longitudinal source
$\inner{f_0}{\partial_\mu a_\mu}$. A transverse probe has no such source and
therefore displays no pole in the extraction map. The singularity belongs to
the reconstruction of a gauge parameter from a displacement that leaves the
Landau slice; it is not a singular response to an arbitrary tangent
displacement.

\section{Local spectral geometry of the horizon}
\label{sec:local}

On the stratum where the lowest eigenvalue is simple, the zero set of that
eigenvalue defines a regular spectral wall whenever its differential does not
vanish. More precisely,
\begin{equation}
  \Hor_{\mathrm{reg}}
  =\{A:\partial\!\cdot A=0,\ \lambda_0(A)=0,\
       \dim\Ker\widehat{\Mfp}[A]=1,\ \dd\lambda_0(A)\ne0\}.
  \label{eq:regular-horizon}
\end{equation}
At a degenerate zero, the ordered lowest eigenvalue is the minimum of several
analytic eigenvalue branches and may therefore fail to be differentiable. The
appropriate first-order datum at such a point is the critical spectral
projector together with the associated matrix crossing form. We consequently
reserve Eq.~\eqref{eq:regular-horizon} for the simple stratum.

To determine the normal covector to this wall, consider the affine transverse
path
\begin{equation}
  A_t=A_\ast+(t-t_\ast)B,
  \qquad
  \partial_\mu B_\mu=0.
  \label{eq:affine-path}
\end{equation}
The corresponding reduced Faddeev-Popov operators have a common domain, and
their dependence on the path parameter is affine:
\begin{equation}
  \widehat{\Mfp}_t
  =\widehat{\Mfp}_{t_\ast}+(t-t_\ast)V_B,
  \qquad
  V_B=-\Pi_0\partial_\mu\ad_{B_\mu}\Pi_0.
  \label{eq:affine-fp-family}
\end{equation}
Choose a simple eigenvalue branch and its normalized eigenvector analytically
near the crossing:
\begin{equation}
  \widehat{\Mfp}_t f_t=\lambda(t)f_t,
  \qquad
  \norm{f_t}=1,
  \label{eq:eigenpair}
\end{equation}
Differentiating this equation and using self-adjointness yields the
Feynman-Hellmann identity
\begin{equation}
  \dot\lambda(t)=\inner{f_t}{V_Bf_t}.
  \label{eq:feynman-hellmann}
\end{equation}
If $f_0$ denotes the normalized zero mode at $t=t_\ast$, the crossing form is
therefore
\begin{align}
  \Gamma_A(B;f_0)
  &=\inner{f_0}{V_Bf_0}
  \notag\\
  &=-\int_M f_0^a\partial_\mu[B_\mu,f_0]^a\,\dd^d x
   =\int_M(\partial_\mu f_0^a)[B_\mu,f_0]^a\,\dd^d x.
  \label{eq:crossing-form-explicit}
\end{align}
The last equality follows by integration by parts under the standing boundary
conditions. The derivative in its final integrand is the ordinary derivative.
Replacing it by $D_A$ introduces an additional term that is generally nonzero
and hence defines a different quadratic form.

For a simple zero, the crossing form gives the leading variation of the
eigenvalue,
\begin{equation}
  \lambda(t)=\Gamma_A(B;f_0)(t-t_\ast)
             +O\bigl((t-t_\ast)^2\bigr).
  \label{eq:linear-crossing}
\end{equation}
Thus $B\mapsto\Gamma_A(B;f_0)$ is the differential of the eigenvalue at the
crossing, and the tangent space to the simple wall is its kernel within the
transverse slice:
\begin{equation}
  T_A\Hor_{\mathrm{reg}}
  =\{B:\partial\!\cdot B=0,\ \Gamma_A(B;f_0)=0\}.
  \label{eq:wall-tangent-space}
\end{equation}
For a rank-$r$ zero, choose an orthonormal basis
$\{f_\alpha\}_{\alpha=1}^r$ of the critical space and define
\begin{equation}
  P=\sum_{\alpha=1}^r\dyad{f_\alpha}{f_\alpha},
  \qquad
  \Gamma_B=PV_BP,
  \qquad
  (\Gamma_B)_{\alpha\beta}
  =\inner{f_\alpha}{V_Bf_\beta}.
  \label{eq:matrix-crossing-form}
\end{equation}
The eigenvalues of $\Gamma_B$ are the first derivatives of the analytic
critical branches along the selected path~\cite{Kato1995}. This matrix, rather
than a derivative of the ordered minimum $\lambda_0$, supplies the regular
first-order description of a degenerate crossing.

\subsection{Concavity and the affine-ray bound}
\label{subsec:concavity}

For each fixed unit vector $f\in H^1\cap\Hilb_0$, the Rayleigh form
\begin{equation}
  q_f(A)=\inner{f}{\widehat{\Mfp}[A]f}
  \label{eq:rayleigh-form}
\end{equation}
depends affinely on the transverse connection. Since the lowest eigenvalue is
the infimum of these affine functions, it is concave on the transverse slice.

\begin{proposition}[Concavity]
\label{prop:concavity}
For transverse $A,A'$ and $0\le\theta\le1$,
\begin{equation}
  \lambda_0\bigl(\theta A+(1-\theta)A'\bigr)
  \ge
  \theta\lambda_0(A)+(1-\theta)\lambda_0(A').
  \label{eq:concavity}
\end{equation}
Consequently, $\Omega$ in Eq.~\eqref{eq:first-region} is convex.
\end{proposition}

\begin{proof}
Affinity gives
\begin{equation}
  q_f\bigl(\theta A+(1-\theta)A'\bigr)
  =\theta q_f(A)+(1-\theta)q_f(A').
  \label{eq:affine-rayleigh}
\end{equation}
Each term on the right is bounded below by the corresponding lowest
eigenvalue. Taking the infimum over normalized $f$ on the left proves
Eq.~\eqref{eq:concavity}.
\end{proof}

Along an affine ray, the same affine dependence gives a stronger statement at
the level of the full critical block.

\begin{theorem}[Negative-definite crossings on affine rays]
\label{thm:affine-ray}
Let
\begin{equation}
  \widehat{\Mfp}_t=\widehat{\Mfp}_0+tV,
  \qquad
  \widehat{\Mfp}_0\ge\lambda_0(0)\one,
  \qquad
  \lambda_0(0)>0.
  \label{eq:positive-ray-family}
\end{equation}
If $t_\ast>0$ and $P$ is the spectral projector onto
$\Ker\widehat{\Mfp}_{t_\ast}$, then
\begin{equation}
  \Gamma=PVP
  =-\frac{1}{t_\ast}P\widehat{\Mfp}_0P
  \le-\frac{\lambda_0(0)}{t_\ast}P.
  \label{eq:matrix-affine-bound}
\end{equation}
Thus $\Gamma$ is negative definite on $\Ran P$. Every zero met along an affine
ray from a configuration in $\Omega$ is a downward, regular crossing in the
radial direction.
\end{theorem}

\begin{proof}
Since $\widehat{\Mfp}_{t_\ast}P=0$,
\begin{equation}
  P\widehat{\Mfp}_0P+t_\ast PVP=0.
  \label{eq:compressed-zero}
\end{equation}
Solving for $PVP$ gives the equality in
Eq.~\eqref{eq:matrix-affine-bound}. The lower bound on
$\widehat{\Mfp}_0$ then gives the stated operator inequality on $\Ran P$.
\end{proof}

For the first zero $t_\ast$ of the lowest eigenvalue, concavity also implies the
one-sided estimate
\begin{equation}
  D^-\lambda_0(t_\ast)
  \le-\frac{\lambda_0(0)}{t_\ast}<0.
  \label{eq:concavity-derivative-bound}
\end{equation}
Equation~\eqref{eq:matrix-affine-bound} contains more information at a
degenerate crossing because it bounds every direction in the critical block.
Neither estimate excludes tangent vectors satisfying
Eq.~\eqref{eq:wall-tangent-space}; rather, both exclude zero velocity in the
affine radial direction specified in the theorem. A nonlinear path, or a path
whose base point lies outside $\Omega$, may instead have tangential contact
with the wall. 

\subsection{Spectral flow and the normal-Hessian index}
\label{subsec:spectral-flow}

We assign positive spectral flow to an eigenvalue that enters the negative
spectrum as the path parameter increases. If a regular crossing has crossing-
form eigenvalues $\gamma_\alpha\ne0$, the corresponding change in the number
of negative eigenvalues is
\begin{equation}
  \Delta n_-
  =-\sum_{\alpha=1}^{r}\sgn\gamma_\alpha.
  \label{eq:spectral-flow}
\end{equation}
For the affine rays covered by Theorem~\ref{thm:affine-ray}, every
$\gamma_\alpha$ is negative, and therefore
\begin{equation}
  \Delta n_-=r.
  \label{eq:downward-index-jump}
\end{equation}
Because $\widehat{\Mfp}$ is the normal Hessian of the norm functional, this
spectral-flow event is simultaneously a jump of $r$ in the normal Morse index.
Together with Theorem~\ref{thm:focal-index}, this identification yields the
local correspondence
\begin{equation}
  \begin{gathered}
  \text{additional FP zero mode}
  \ \longleftrightarrow\
  \text{additional focal direction}
  \\
  \longleftrightarrow\
  \text{normal-index jump}
  \ \longleftrightarrow\
  \text{spectral-flow event}
  \end{gathered}.
  \label{eq:local-dictionary}
\end{equation}
Our sign convention in Eq.~\eqref{eq:spectral-flow} is the negative of the
convention often adopted in the mathematical spectral-flow literature
~\cite{RobbinSalamon1995}.

\section{Resolvent normal forms}
\label{sec:resolvent}

The crossing form also fixes the leading singular part of the reduced
resolvent. Let $\delta=t-t_\ast$, let $f_0$ be the normalized zero mode, and set
$P=\dyad{f_0}{f_0}$ and $Q=\one-P$. With respect to the fixed decomposition
$\Ran P\oplus\Ran Q$ at $t=t_\ast$, the operator blocks satisfy
\begin{align}
  P\widehat{\Mfp}_tP
  &=\delta\Gamma P+O(\delta^2),
  \label{eq:simple-critical-block}\\
  Q\widehat{\Mfp}_tP
  &=\delta QV_BP+O(\delta^2),
  \label{eq:simple-offdiag}\\
  Q\widehat{\Mfp}_tQ
  &=Q\widehat{\Mfp}_{t_\ast}Q+O(\delta).
  \label{eq:simple-complement}
\end{align}
Because the zero eigenvalue is simple and isolated, the last block is separated
from zero and has a bounded inverse in a neighborhood of the crossing. The
Schur complement on the critical subspace is then
\begin{equation}
  S(\delta)
  =P\widehat{\Mfp}_tP
   -P\widehat{\Mfp}_tQ
    (Q\widehat{\Mfp}_tQ)^{-1}
    Q\widehat{\Mfp}_tP
  =\delta\Gamma P+O(\delta^2).
  \label{eq:simple-schur}
\end{equation}
Both off-diagonal factors are of order $\delta$, so their self-energy
contribution starts at order $\delta^2$ and does not modify the residue of the
leading pole. Block inversion therefore gives
\begin{equation}
  \widehat{\Mfp}_t^{-1}
  =\frac{P}{\delta\Gamma}+O(1).
  \label{eq:simple-laurent}
\end{equation}
For a source $J_t$ that is continuous at $t=t_\ast$, only its projection onto
the critical mode contributes to the singular term:
\begin{equation}
  \widehat{\Mfp}_t^{-1}J_t
  =\frac{\ket{f_0}\inner{f_0}{J_{t_\ast}}}
         {\Gamma(t-t_\ast)}+O(1).
  \label{eq:sourced-simple-pole}
\end{equation}
Equation~\eqref{eq:oblique-pole-preview} follows by taking
$J=-\partial\!\cdot a$. In particular, the inverse operator has a pole even
when a chosen source does not detect it; the singular contribution to the
sourced solution vanishes precisely when the source is orthogonal to $f_0$.

\subsection{Exact rank-one Feshbach equation}
\label{subsec:rank-one-feshbach}

The same cancellation can be expressed through an exact scalar Feshbach
equation without following the moving eigenvector. Freeze $f_0$ at the crossing
and define
\begin{equation}
  v_t=Q\widehat{\Mfp}_tf_0.
  \label{eq:feshbach-v}
\end{equation}
Whenever $Q(\widehat{\Mfp}_t-z)Q$ is invertible, eliminating the complementary
component reduces the eigenvalue problem to
\begin{equation}
  h(t,z)=z,
  \qquad
  h(t,z)=
  \inner{f_0}{\widehat{\Mfp}_tf_0}
  -\inner{v_t}{
      (Q\widehat{\Mfp}_tQ-z)^{-1}v_t}.
  \label{eq:feshbach-scalar}
\end{equation}
Implicit differentiation along a critical root $z=\lambda(t)$ yields
\begin{equation}
  \dot\lambda(t)
  =\frac{\partial_t h(t,\lambda(t))}
         {1-\partial_z h(t,\lambda(t))}.
  \label{eq:feshbach-derivative}
\end{equation}
At threshold, $v_{t_\ast}=0$ because $f_0$ lies in the kernel of the full
operator. Consequently,
\begin{equation}
  \partial_z h(t_\ast,0)=0,
  \qquad
  \partial_t h(t_\ast,0)=\Gamma,
  \label{eq:feshbach-threshold}
\end{equation}
so the frozen-state spectral-weight factor in the denominator of
Eq.~\eqref{eq:feshbach-derivative} equals one at the crossing. Away from
threshold,
\begin{equation}
  \partial_z h(t,z)
  =-\norm{(Q\widehat{\Mfp}_tQ-z)^{-1}v_t}^2\le0.
  \label{eq:feshbach-z-derivative}
\end{equation}
These identities describe the local algebra of a spectral projection. They
should not be identified with the nonlocal horizon functional that enters the
quantum measure.

\subsection{Regular degenerate crossings}
\label{subsec:regular-degenerate}

The first Schur reduction extends directly to a critical subspace of dimension
$r$. Suppose $\dim\Ker\widehat{\Mfp}_{t_\ast}=r$, let $P$ project onto that
kernel, and set $Q=\one-P$. Eliminating the complementary subspace gives the
effective matrix
\begin{equation}
  E(t,z)
  =P(\widehat{\Mfp}_t-z)P
  -P\widehat{\Mfp}_tQ
   \bigl[Q(\widehat{\Mfp}_t-z)Q\bigr]^{-1}
   Q\widehat{\Mfp}_tP.
  \label{eq:degenerate-feshbach}
\end{equation}
At $(t,z)=(t_\ast,0)$, the off-diagonal couplings vanish. Hence, with
$\delta=t-t_\ast$,
\begin{equation}
  E(t,0)=\delta\Gamma+O(\delta^2),
  \qquad
  \Gamma=P\dot{\widehat{\Mfp}}_{t_\ast}P.
  \label{eq:degenerate-leading}
\end{equation}
If $\Gamma$ is invertible on $\Ran P$, the crossing is regular in every
critical direction, and block inversion gives the operator-valued residue
\begin{equation}
  \widehat{\Mfp}_t^{-1}
  =\frac{1}{t-t_\ast}P\Gamma^{-1}P+O(1).
  \label{eq:regular-degenerate-residue}
\end{equation}
The residue operator has rank $r$. For a fixed source,
\begin{equation}
  \widehat{\Mfp}_t^{-1}J_t
  =\frac{1}{t-t_\ast}P\Gamma^{-1}PJ_{t_\ast}+O(1),
  \label{eq:regular-degenerate-source}
\end{equation}
and the singular term vanishes exactly when $PJ_{t_\ast}=0$. Here $PJ$ is a
vector in the critical space; the rank-$r$ statement refers to the residue
operator.

Every crossing covered by Theorem~\ref{thm:affine-ray} is regular in this sense
along its radial path. Path regularity is not solely a property of the endpoint,
however: another path through the same configuration may be tangent to the
wall, and a general analytic family may satisfy $\Ker\Gamma\ne\{0\}$.

\subsection{Tangential paths and the second Schur reduction}
\label{subsec:tangential}

When $\Gamma$ has a kernel, replacing $\Gamma^{-1}$ with its Moore-Penrose
pseudoinverse does not recover the full singular resolvent. The diagonal family
\begin{equation}
  E(\delta)=
  \begin{pmatrix}
    \delta&0\\
    0&\delta^2
  \end{pmatrix}
  \label{eq:pseudoinverse-counterexample}
\end{equation}
already demonstrates the failure: although
$\Gamma=\operatorname{diag}(1,0)$, the inverse in the second component has a
double pole that the pseudoinverse of $\Gamma$ cannot encode.

Off-diagonal coupling within the critical space introduces an additional
effect. Let $P_0$ project onto $\Ker\Gamma$ and set $P_1=P-P_0$. In the
decomposition $\Ran P=\Ran P_1\oplus\Ran P_0$, write the effective operator
$E(\delta)\equiv E(t_\ast+\delta,0)$ as
\begin{equation}
  E(\delta)=
  \begin{pmatrix}
    E_{11}(\delta)&E_{10}(\delta)\\
    E_{01}(\delta)&E_{00}(\delta)
  \end{pmatrix}.
  \label{eq:critical-two-block}
\end{equation}
Since the first derivative is $\Gamma$, its restriction to $\Ran P_1$ is
invertible, whereas all first-order matrix elements involving $\Ran P_0$
vanish. It follows that
\begin{equation}
  E_{11}(\delta)=\delta\Gamma_1+O(\delta^2),
  \quad
  E_{10}(\delta)=O(\delta^2),
  \quad
  E_{00}(\delta)=O(\delta^2),
  \label{eq:critical-block-orders}
\end{equation}
where $\Gamma_1=\Gamma|_{\Ran P_1}$ is invertible. After the complementary
ghost space has been eliminated in Eq.~\eqref{eq:degenerate-feshbach}, a second
Schur reduction removes the linearly crossing part of the critical space. The
operator governing the tangential sector is
\begin{equation}
  F_0(\delta)
  =E_{00}(\delta)
   -E_{01}(\delta)E_{11}(\delta)^{-1}E_{10}(\delta).
  \label{eq:second-schur}
\end{equation}
Because $E_{01}$ and $E_{10}$ start at order $\delta^2$, while
$E_{11}^{-1}$ starts at order $\delta^{-1}$, the correction in
Eq.~\eqref{eq:second-schur} begins at order $\delta^3$. It can therefore alter
the leading coefficient, and may alter the pole order, whenever the direct
$P_0$ compression begins at third or higher order.

\begin{proposition}[Tangential pole order]
\label{prop:tangential-pole}
Assume that, after the two reductions above,
\begin{equation}
  F_0(\delta)
  =\delta^mK_m+O(\delta^{m+1}),
  \qquad
  m\ge2,
  \qquad
  K_m:\Ran P_0\to\Ran P_0\ \text{invertible}.
  \label{eq:tangential-leading}
\end{equation}
Then
\begin{equation}
  P_0\widehat{\Mfp}_{t_\ast+\delta}^{-1}P_0
  =\delta^{-m}K_m^{-1}
   +O(\delta^{-(m-1)}).
  \label{eq:tangential-pole}
\end{equation}
\end{proposition}

\begin{proof}
The $P_0P_0$ block of the inverse of Eq.~\eqref{eq:critical-two-block} is
$F_0(\delta)^{-1}$. Factoring
$F_0(\delta)=\delta^mK_m[\one+O(\delta)]$ and applying a Neumann expansion
proves Eq.~\eqref{eq:tangential-pole}.
\end{proof}

The finite-dimensional family
\begin{equation}
  E(\delta)=
  \begin{pmatrix}
    \delta&\delta^2\\
    \delta^2&\delta^4
  \end{pmatrix}
  \label{eq:second-schur-example}
\end{equation}
shows why the second reduction cannot be omitted. Although the direct
tangential compression is $E_{00}=\delta^4$, the Schur complement is
$F_0=\delta^4-\delta^3=-\delta^3[1-\delta]$. The tangential inverse block thus
has a cubic rather than a quartic pole. If distinct tangential branches vanish
to different orders, the same procedure applies after resolving the associated
analytic subspaces; the collection of branch orders then replaces the single
exponent $m$.

\section{Symmetry, degeneracy, and expected codimension}
\label{sec:symmetry}

The matrix crossing form inherits every symmetry preserved by the operator
family along the path. Let a compact group $G$ act orthogonally on the reduced
ghost space and suppose
\begin{equation}
  U(g)\widehat{\Mfp}_tU(g)^{-1}=\widehat{\Mfp}_t
  \quad\text{for all }g\in G
  \label{eq:symmetry-commutation}
\end{equation}
near the crossing. Here $G$ denotes a symmetry group of the background and the
operator; it need not coincide with a gauge stabilizer.

\begin{lemma}[Real symmetric Schur lemma]
\label{lem:real-schur}
Let $V$ be a finite-dimensional real inner-product space carrying an
irreducible orthogonal representation of $G$. If a symmetric operator
$X:V\to V$ commutes with the representation, then $X=x\one_V$ for some real
$x$.
\end{lemma}

\begin{proof}
By the real spectral theorem, $V$ is the orthogonal sum of the eigenspaces of
$X$. Commutation with the representation makes each eigenspace invariant under
$G$. Any chosen eigenspace is nonzero, and irreducibility forces it to equal
all of $V$; hence $X$ has a single eigenvalue and is a scalar operator.
\end{proof}

\begin{theorem}[Symmetry-protected critical multiplet]
\label{thm:symmetry-protected}
Under Eq.~\eqref{eq:symmetry-commutation}, suppose the real critical space
$\Ran P$ is irreducible under $G$. Then
\begin{equation}
  \Gamma=P\dot{\widehat{\Mfp}}_{t_\ast}P
  =\gamma P.
  \label{eq:scalar-crossing-form}
\end{equation}
If $\gamma\ne0$, the singular resolvent is
\begin{equation}
  \widehat{\Mfp}_t^{-1}
  =\frac{P}{\gamma(t-t_\ast)}+O(1).
  \label{eq:symmetric-residue}
\end{equation}
For an affine ray from $\Omega$, Theorem~\ref{thm:affine-ray} gives
$\gamma\le-\lambda_0(0)/t_\ast<0$.
\end{theorem}

\begin{proof}
Equation~\eqref{eq:symmetry-commutation} implies that the spectral projector
$P$ commutes with $U(g)$, and differentiating that equation shows that
$\dot{\widehat{\Mfp}}_{t_\ast}$ does as well. Their compression to $\Ran P$ is
symmetric and intertwines the restricted representation. Lemma~\ref{lem:real-schur}
therefore proves Eq.~\eqref{eq:scalar-crossing-form}, while
Eq.~\eqref{eq:symmetric-residue} follows from
Eq.~\eqref{eq:regular-degenerate-residue}.
\end{proof}

For an unrestricted real symmetric $r\times r$ critical block, requiring an
$r$-fold eigenvalue to lie at zero imposes $r(r+1)/2$ real conditions: every
traceless component must vanish, and the remaining scalar eigenvalue must also
vanish~\cite{vonNeumannWigner1929}. In a Faddeev-Popov family, this dimension
count is only an \emph{expected generic codimension}. It applies when the
compressed perturbation map is transverse,
\begin{equation}
  B\longmapsto P\,\delta_B\widehat{\Mfp}[A]\,P
  \quad\text{into }\operatorname{Sym}(\Ran P).
  \label{eq:compressed-perturbation-map}
\end{equation}
Locality, the transversality constraint on $B$, and representation-theoretic
selection rules can reduce the image of this map, so the unrestricted matrix
count is not by itself a codimension theorem for the space of gauge fields.
Within a symmetry-preserving family for which $\Ran P$ is irreducible,
Eq.~\eqref{eq:scalar-crossing-form} leaves only one scalar parameter. A whole
critical multiplet may then cross zero on a codimension-one wall of that
restricted family.

Selection rules can protect degeneracy without making the critical space
irreducible. Conserved abelian quantum numbers may force the compressed matrix
to be block diagonal, whereas a discrete exchange symmetry may require the
block eigenvalues to coincide. The resulting critical space can be reducible
even when the full crossing matrix is proportional to the identity. The
examples below realize both mechanisms.

\section{A projected single-harmonic test}
\label{sec:mathieu}

Consider the transverse $SU(2)$ background on the $2\pi$-periodic torus
$T^2$,
\begin{equation}
  A_\mu^a(x)
  =a\,\delta^{a3}\delta_{\mu2}\cos x_1.
  \label{eq:single-harmonic-background}
\end{equation}
This family is reducible: the constant Cartan generator $T^3$ is covariantly
constant, and the neutral color sector remains free. Its charged sector is
nevertheless a useful finite-volume realization of the crossing formulas,
provided that the projection and the residual degeneracy are kept explicit.

Combining the two real charged ghost components into
$\phi=\psi^1+\iu\psi^2$ and using the Fourier expansion
\begin{equation}
  \phi(x_1,x_2)
  =\sum_{p,q\in\mathbb Z}\phi_{p,q}
   e^{\iu(px_1+qx_2)}
  \label{eq:charged-fourier}
\end{equation}
gives, in one charged-helicity convention,
\begin{equation}
  (\Mfp\phi)_{p,q}
  =(p^2+q^2)\phi_{p,q}
  +\frac{aq}{2}\bigl(\phi_{p-1,q}+\phi_{p+1,q}\bigr).
  \label{eq:mathieu-block}
\end{equation}
Thus, for each fixed $q$, the charged Faddeev-Popov operator reduces to a
tridiagonal Mathieu-type operator on the $p$ ladder. The opposite charged
convention is obtained by replacing $a$ with $-a$.

The smallest symmetric truncation to $q=1$ and $p=0,\pm1$ gives
\begin{equation}
  M_3=
  \begin{pmatrix}
    1&a/2&a/2\\
    a/2&2&0\\
    a/2&0&2
  \end{pmatrix}.
  \label{eq:three-state}
\end{equation}
The antisymmetric combination of the $p=\pm1$ states decouples, whereas the
even subspace carries
\begin{equation}
  M_{\mathrm{even}}=
  \begin{pmatrix}
    1&a/\sqrt2\\
    a/\sqrt2&2
  \end{pmatrix},
  \label{eq:two-state-even}
\end{equation}
whose lower eigenvalue vanishes at $a=2$. Symmetric truncation of the full
ladder to $\abs{p}\le N$ shifts this root to a stable limiting value:
$a_\ast(3)=1.9053248$, while $a_\ast(5)$ and $a_\ast(40)$ differ by only
$1.7\times10^{-11}$. At $N=40$, the critical data in the projected complex
$q=1$ block, quoted to six decimal places, are
\begin{equation}
  a_\ast=1.905316,
  \qquad
  \lambda_1(a_\ast)-\lambda_0(a_\ast)=1.714198,
  \qquad
  \Gamma=\dot\lambda_0(a_\ast)=-0.737092.
  \label{eq:mathieu-values}
\end{equation}
The nonzero gap establishes that the null vector is simple within this
projected block. It does not establish simplicity of the full real charged
operator.

The remaining momentum sectors can be controlled without further
diagonalization. If $S$ denotes the bilateral shift on the $p$ ladder, then
$\norm{S+S^\ast}\le2$, and Eq.~\eqref{eq:mathieu-block} implies the quadratic-form
bound
\begin{equation}
  \Mfp_q
  \ge q^2-\abs{aq}
  =\abs{q}\bigl(\abs{q}-\abs{a}\bigr).
  \label{eq:q-sector-bound}
\end{equation}
Because $a_\ast<2$, all sectors with $\abs{q}\ge2$ are strictly positive at the
crossing. The $q=0$ block is free once its constant mode has been removed, so
only $q=\pm1$ can be critical. A simple complex null vector contributes two
real charged directions in each of these momentum sectors; consequently, the
full real charged critical space has dimension four. Translation in $x_2$
prevents mixing between $q=1$ and $q=-1$, Cartan rotations act within each real
charged doublet, and the discrete charge-momentum symmetry equates the two
block velocities. The degeneracy is therefore protected by selection rules,
but it belongs to a reducible background and falls outside the irreducibility
hypothesis of Theorem~\ref{thm:focal-index}.

Figure~\ref{fig:mathieu} displays the crossing and its slope. The
Feynman-Hellmann matrix element was compared with a centered five-point
derivative of the computed eigenvalue over the plotted range. This comparison
checks the implementation of the projected matrix problem; it supplies no
independent evidence for the full gauge-orbit geometry.

\begin{figure}[t]
  \centering
  \includegraphics[width=\linewidth]{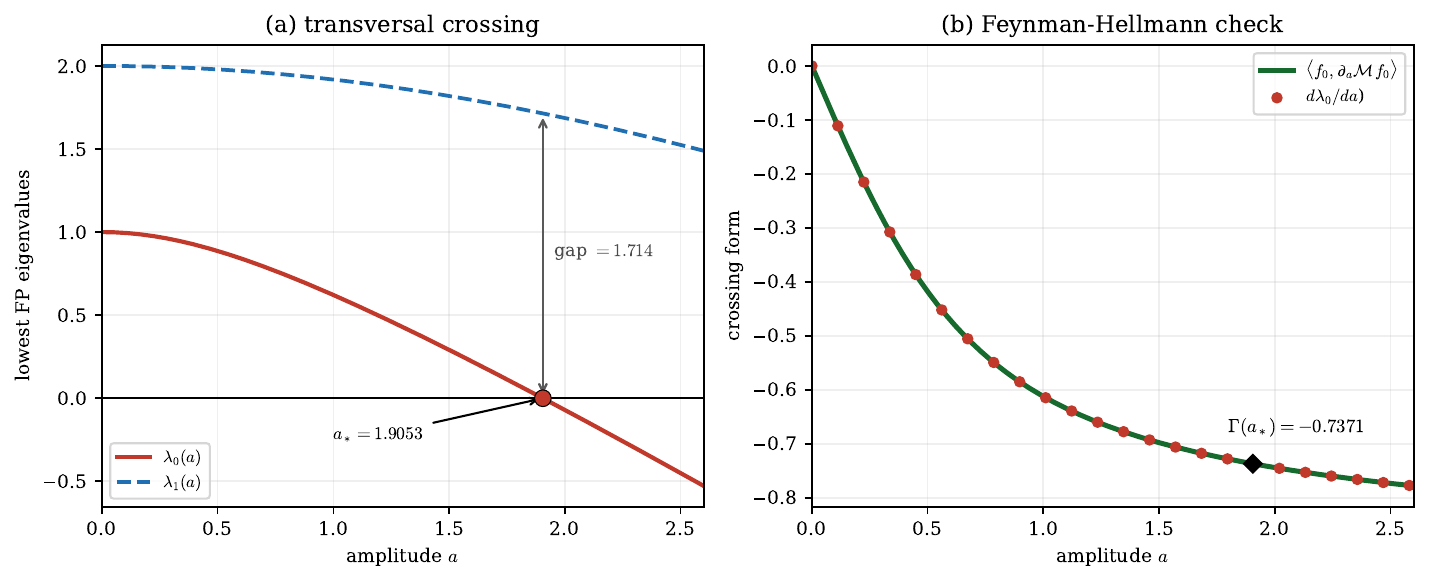}
  \caption{Projected single-harmonic calculation in the complex $q=1$ block.
  Panel (a) shows the zero of the lowest eigenvalue at
  $a_\ast\simeq1.905316$ and the nonzero gap to the next eigenvalue in the
  same block. Panel (b) compares the Feynman-Hellmann matrix element
  $\inner{f_0}{(\partial_a\Mfp)f_0}$ with a centered five-point derivative of
  $\lambda_0(a)$; at the crossing, $\Gamma\simeq-0.737092$. The calculation is
  finite dimensional after the $p$-ladder cutoff and concerns a reducible,
  sector-projected background.}
  \label{fig:mathieu}
\end{figure}

\section{The \texorpdfstring{$SU(2)$}{SU(2)} hedgehog and continuum thresholds}
\label{sec:hedgehog}

We next consider the radial family on $\mathbb R^3$,
\begin{equation}
  A_i^c(x)=\varepsilon_{cij}x_jh(r),
  \qquad
  r=\abs{\bm x}.
  \label{eq:hedgehog-background}
\end{equation}
Antisymmetry of $\varepsilon_{cij}$ gives transversality for every differentiable
radial profile:
\begin{equation}
  \partial_iA_i^c
  =\varepsilon_{cij}\delta_{ij}h
   +\varepsilon_{cij}x_jx_i\frac{h'}{r}
  =0.
  \label{eq:hedgehog-transverse}
\end{equation}
For the decaying profile used below, each radial operator approaches the free
radial Laplacian at infinity and has essential spectrum beginning at zero. The
normalizable zero modes are therefore threshold eigenstates rather than
isolated eigenvalues separated by a positive gap. The compact-manifold
resolvent theorems established above cannot be transferred verbatim to this
continuum problem. We use the continuum family to obtain exact threshold data
and compare those data with discrete eigenvalue crossings in a Dirichlet radial
box. The box supports the spectral statements for its one-dimensional
Sturm-Liouville operators, but, in the absence of boundary conditions for the
full gauge field and gauge group, it is not a realization of the complete
Singer geometry discussed in Sections~\ref{sec:setting} and \ref{sec:local}.

Introduce the color-spin and orbital generators
\begin{equation}
  (S_i)_{ac}=-\iu\varepsilon_{iac},
  \qquad
  L_i=-\iu\varepsilon_{ijk}x_j\partial_k.
  \label{eq:spin-orbit-generators}
\end{equation}
Substitution of Eq.~\eqref{eq:hedgehog-background} into the Faddeev-Popov
vertex yields
\begin{equation}
  \varepsilon^{abc}A_i^b\partial_i\omega^c
  =-h(r)(\bm S\!\cdot\!\bm L\,\omega)^a,
  \label{eq:hedgehog-vertex}
\end{equation}
and hence
\begin{equation}
  \Mfp=-\partial^2+h(r)\bm S\!\cdot\!\bm L.
  \label{eq:hedgehog-normal-form}
\end{equation}
This operator commutes with $\bm J=\bm L+\bm S$, $\bm L^2$, and the radial
Hamiltonian. On a channel labeled by $(J,L)$,
\begin{equation}
  \bm S\!\cdot\!\bm L
  =c_{JL}
  =\frac12\bigl[J(J+1)-L(L+1)-2\bigr],
  \label{eq:cJL}
\end{equation}
with angular multiplicity $2J+1$. Writing the radial wave function as
$R(r)=u(r)/r$ reduces the eigenvalue equation to
\begin{equation}
  -u''(r)+\frac{L(L+1)}{r^2}u(r)
  +c_{JL}h(r)u(r)=\lambda u(r).
  \label{eq:radial-channel}
\end{equation}
The channels used below are collected in Table~\ref{tab:channels}. The
$(J,L)=(1,0)$ channel has $c_{10}=0$ and is consequently independent of the
hedgehog amplitude.

\begin{table}[t]
\centering
\caption{Selected channels of the hedgehog Faddeev-Popov operator. Parity is
$(-1)^L$, the angular multiplicity is $2J+1$, and the last column is the radial
centrifugal term in Eq.~\eqref{eq:radial-channel}.}
\label{tab:channels}
\begin{tabular}{ccccc}
\toprule
$J^P$ & $L$ & $c_{JL}$ & multiplicity & centrifugal term\\
\midrule
$0^-$ & $1$ & $-2$ & $1$ & $2/r^2$\\
$1^+$ & $0$ & $0$  & $3$ & $0$\\
$1^-$ & $1$ & $-1$ & $3$ & $2/r^2$\\
$1^+$ & $2$ & $-3$ & $3$ & $6/r^2$\\
$2^-$ & $1$ & $+1$ & $5$ & $2/r^2$\\
$2^+$ & $2$ & $-1$ & $5$ & $6/r^2$\\
$2^-$ & $3$ & $-4$ & $5$ & $12/r^2$\\
\bottomrule
\end{tabular}
\end{table}

The diagonal $SO(3)$ action that rotates space and color simultaneously is an
operator symmetry; it is not a gauge stabilizer. Irreducibility of a nonzero
hedgehog background follows independently from the integrability condition for
a covariantly constant adjoint scalar. If $D_i\xi=0$, then
\begin{equation}
  [F_{ij},\xi]=0.
  \label{eq:irreducibility-integrability}
\end{equation}
On the positive $x_3$ axis, the curvature contains the color components
\begin{equation}
  F_{12}^3=-2h+r^2h^2,
  \qquad
  F_{23}^1=F_{31}^2=-(2h+rh').
  \label{eq:hedgehog-curvature-axis}
\end{equation}
For the profile in Eq.~\eqref{eq:henyey-profile} and $g\ne0$, these coefficients
are simultaneously nonzero on a sufficiently small interval with $r>0$.
There the curvature spans all three color directions, so
Eq.~\eqref{eq:irreducibility-integrability} forces $\xi$ to vanish. Covariant
constancy then extends the vanishing throughout the connected domain. The
nonzero threshold backgrounds considered below are therefore irreducible,
whereas $g=0$ is the reducible free connection.

\subsection{A closed-form threshold state in every coupled channel}
\label{subsec:henyey}

Set the core scale to unity and take
\begin{equation}
  h(r)=g\phi(r),
  \qquad
  \phi(r)=\frac{9r}{(r^3+1)^2}.
  \label{eq:henyey-profile}
\end{equation}
The coefficient in $\phi(r)$ fixes the coupling normalization; the channel
minimization below then places the first positive threshold at $g=1$. Following
Henyey's construction, we prescribe a radial mode and solve the zero-energy
equation for the background~\cite{Henyey1979}. Consider
\begin{equation}
  u_L(r)=\frac{r^{L+1}}{(r^3+1)^p}.
  \label{eq:henyey-ansatz}
\end{equation}
Direct differentiation gives
\begin{equation}
  \frac{u_L''}{u_L}-\frac{L(L+1)}{r^2}
  =
  \frac{-6p(L+2)r(r^3+1)+9p(1+p)r^4}
       {(r^3+1)^2}.
  \label{eq:henyey-log-curvature}
\end{equation}
Matching this radial dependence to Eq.~\eqref{eq:henyey-profile} requires the
$r^4$ coefficient to vanish, which fixes
\begin{equation}
  p=\frac{2L+1}{3}.
  \label{eq:henyey-p}
\end{equation}
The remaining terms then reduce to the identity
\begin{equation}
  \frac{u_L''}{u_L}-\frac{L(L+1)}{r^2}
  =-2(2L+1)(L+2)\frac{r}{(r^3+1)^2}.
  \label{eq:henyey-identity}
\end{equation}
At zero energy, the left-hand side equals
$c_{JL}g\phi(r)=9gc_{JL}r/(r^3+1)^2$. Consequently, every coupled channel
$c_{JL}\ne0$ has the explicit threshold state and amplitude
\begin{equation}
  u_L(r)=
  \frac{r^{L+1}}{(r^3+1)^{(2L+1)/3}},
  \qquad
  g_\ast(J,L)=
  -\frac{2(2L+1)(L+2)}{9c_{JL}}.
  \label{eq:henyey-zero-and-wall}
\end{equation}
For $L\ge1$, the behavior $u_L(r)\sim r^{L+1}$ near the origin and
$u_L(r)\sim r^{-L}$ at infinity places the state in
$L^2(\mathbb R_+,\dd r)$. Equation~\eqref{eq:henyey-zero-and-wall} therefore
identifies one normalizable threshold state in every coupled channel but does
not determine the remainder of the radial spectrum.

Because $u_L$ is strictly positive for $0<r<\infty$, it has no interior node.
The Sturm oscillation theorem identifies it as the radial ground state at the
threshold amplitude~\cite{Zettl2005}. Along the signed direction for which
$gc_{JL}<0$, the radial quadratic form decreases monotonically with the
magnitude of the coupling. Hence no negative radial eigenvalue can precede this
nodeless zero-energy state, and $g_\ast(J,L)$ is the first threshold in its
channel.

A core parameter $\alpha>0$ may be restored by replacing $r^3+1$ with
$r^3+\alpha$. The same differentiation gives
\begin{equation}
  \frac{u_L''}{u_L}-\frac{L(L+1)}{r^2}
  =-2(2L+1)(L+2)\alpha
    \frac{r}{(r^3+\alpha)^2}.
  \label{eq:henyey-alpha}
\end{equation}
Under $r=\alpha^{1/3}\rho$, the dimensionless coupling is $g/\alpha$, and
$g_\ast^{(\alpha)}=\alpha g_\ast^{(1)}$. The choice $\alpha=1$ therefore fixes
the radial unit without restricting the family.

For color spin $S=1$ and a fixed orbital angular momentum $L\ge1$, the allowed
total angular momenta are $J=L-1,L,L+1$. Equation~\eqref{eq:cJL} gives
\begin{equation}
  c_{L-1,L}=-(L+1),
  \qquad
  c_{L,L}=-1,
  \qquad
  c_{L+1,L}=L.
  \label{eq:three-c-families}
\end{equation}
Substitution into Eq.~\eqref{eq:henyey-zero-and-wall} separates the thresholds
into three families:
\begin{align}
  g_\ast^{J=L-1}
  &=\frac{2(2L+1)(L+2)}{9(L+1)},
  \label{eq:g-family-minus}\\
  g_\ast^{J=L}
  &=\frac{2(2L+1)(L+2)}{9},
  \label{eq:g-family-zero}\\
  g_\ast^{J=L+1}
  &=-\frac{2(2L+1)(L+2)}{9L}.
  \label{eq:g-family-plus}
\end{align}
The first two families are attractive for $g>0$, whereas the third is
attractive for $g<0$. Their absolute values increase with $L$. For the first
and third families, this follows from the decompositions
\begin{equation}
  \frac{(2L+1)(L+2)}{L+1}
  =2L+3-\frac{1}{L+1},
  \qquad
  \frac{(2L+1)(L+2)}{L}
  =2L+5+\frac{2}{L},
  \label{eq:family-monotonicity}
\end{equation}
whose discrete increments are positive for $L\ge1$; monotonicity of the middle
family follows directly from its quadratic numerator. The thresholds closest
to the free connection are therefore
\begin{equation}
  \min_{L\ge1}g_\ast^{J=L-1}=1,
  \qquad
  \min_{L\ge1}g_\ast^{J=L}=2,
  \qquad
  \max_{L\ge1}g_\ast^{J=L+1}=-2.
  \label{eq:family-extrema}
\end{equation}
The first positive threshold is the $(J,L)=(0,1)$ singlet at $g=1$, and the
first negative threshold is the $(J,L)=(2,1)$ quintet at $g=-2$. Repulsive
channels cannot produce a negative eigenvalue, and the nodeless-state argument
excludes an earlier threshold in every attractive channel. Every radial channel
is therefore nonnegative throughout $-2\le g\le1$, with normalizable
zero-energy threshold states at the two endpoints. We call the open interval
between them the threshold-stability interval:
\begin{equation}
  \mathcal I_{\mathrm{thr}}=(-2,1).
  \label{eq:threshold-interval}
\end{equation}
Because zero is the bottom of the essential spectrum, this is a threshold
stability interval. It is not a positive-gap counterpart of the compact region
$\Omega$.

\Needspace{0.42\textheight}
Table~\ref{tab:henyey-walls} records representative members of the three
families. The table is illustrative; the interval in
Eq.~\eqref{eq:threshold-interval} follows from the minimization over all
$L\ge1$, rather than from this finite list.

\begin{table}[htbp]
\centering
\caption{Selected closed-form thresholds of the hedgehog family. The positive
endpoint of Eq.~\eqref{eq:threshold-interval} is the $0^-$ singlet, and the
negative endpoint is the $2^-$ quintet.}
\label{tab:henyey-walls}
\begin{tabular}{cccccc}
\toprule
$J^P$ & $L$ & $c_{JL}$ & $g_\ast$ & $2J+1$ & $u_L(r)$\\
\midrule
$0^-$ & $1$ & $-2$ & $1$ & $1$ & $r^2/(r^3+1)$\\
$1^+$ & $2$ & $-3$ & $40/27$ & $3$ &
  $r^3/(r^3+1)^{5/3}$\\
$2^-$ & $3$ & $-4$ & $35/18$ & $5$ &
  $r^4/(r^3+1)^{7/3}$\\
$1^-$ & $1$ & $-1$ & $2$ & $3$ & $r^2/(r^3+1)$\\
$2^+$ & $2$ & $-1$ & $40/9$ & $5$ &
  $r^3/(r^3+1)^{5/3}$\\
$2^-$ & $1$ & $+1$ & $-2$ & $5$ & $r^2/(r^3+1)$\\
\bottomrule
\end{tabular}
\end{table}

\subsection{Dirichlet-box comparison and convergence}
\label{subsec:radial-numerics}

To connect the continuum threshold problem with a discrete finite-volume
spectrum, we impose $u(0)=u(R)=0$ on Eq.~\eqref{eq:radial-channel}. A uniform
grid and the three-point approximation to $-u''$ produce a real symmetric
tridiagonal matrix. Figure~\ref{fig:hedgehog-walls} uses
\begin{equation}
  R=40,
  \qquad
  \Delta r=0.02,
  \label{eq:radial-figure-grid}
\end{equation}
for which the lowest-eigenvalue crossings in the singlet, quintet, and triplet
channels occur at
\begin{equation}
  g_{\ast,R}(0^-)=0.999975,
  \qquad
  g_{\ast,R}(2^-)=-1.999950,
  \qquad
  g_{\ast,R}(1^+)=1.481393.
  \label{eq:boxed-roots}
\end{equation}
These values are finite-box approximations to $1$, $-2$, and $40/27$,
respectively. The shaded interval in the figure is instead fixed by the exact
all-channel minimization in Eq.~\eqref{eq:family-extrema}.

\begin{figure}[t]
  \centering
  \includegraphics[width=0.9\linewidth]{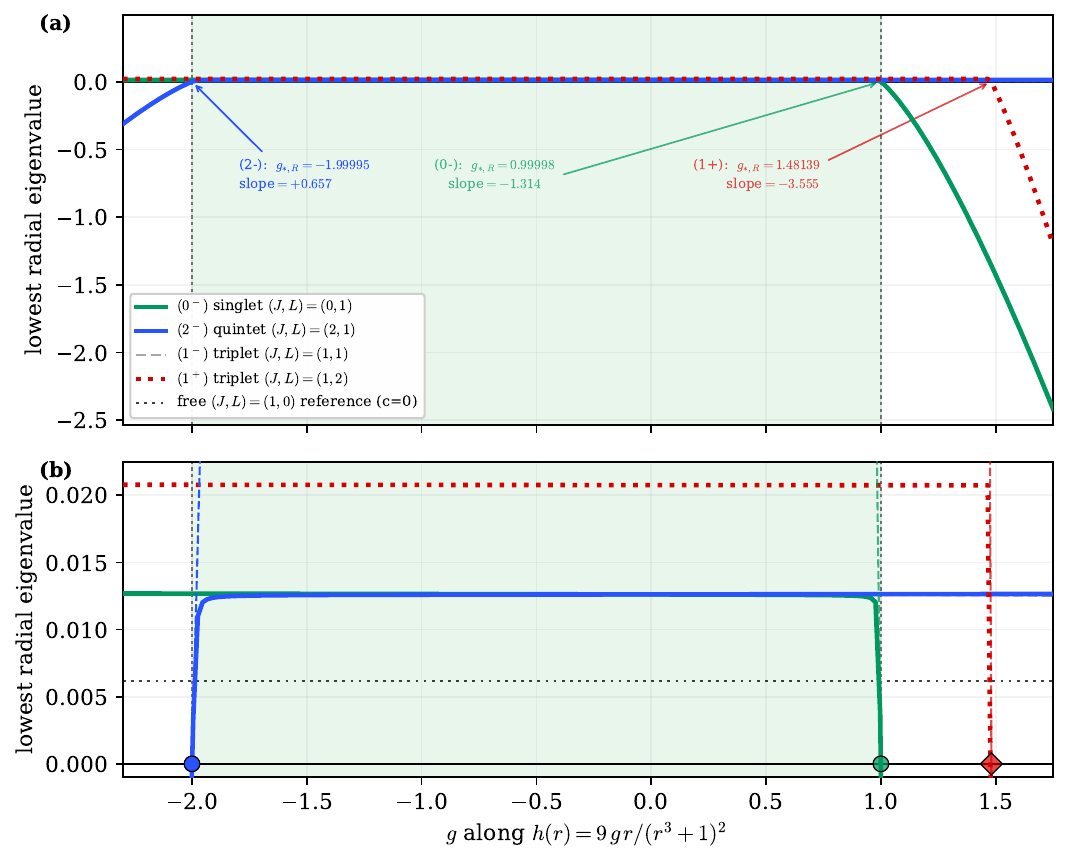}
  \caption{Selected hedgehog channels in a Dirichlet radial box with $R=40$
  and $\Delta r=0.02$. The curves and short tangent segments are finite-box
  eigenvalues and centered finite-difference estimates of the
  Feynman-Hellmann slopes. The shaded interval $-2<g<1$ is the
  independent continuum result obtained by minimizing over all spin-orbit
  channels. The boxed singlet, quintet, and triplet roots approach $1$, $-2$,
  and $40/27$, while the uncoupled $(J,L)=(1,0)$ channel remains free. The plot
  is a selected-channel comparison, not a diagonalization of an unrestricted
  angular operator or a construction of finite-domain gauge-orbit geometry.}
  \label{fig:hedgehog-walls}
\end{figure}

\Needspace{0.52\textheight}
Table~\ref{tab:hedgehog-convergence} separates the two numerical limits. The
first three rows vary $R$ at fixed $\Delta r=0.02$; the last two rows, together
with the shared $(R,\Delta r)=(40,0.02)$ row, vary the mesh at fixed $R=40$.
This arrangement prevents a simultaneous change in $R$ and $\Delta r$ from
being interpreted as evidence for either limit separately.

\begin{table}[htbp]
\centering
\caption{Controlled finite-box comparison of the singlet, quintet, and
triplet wall amplitudes. The radius scan holds $\Delta r$ fixed, and the mesh
scan holds $R$ fixed; the $(40,0.02)$ entry belongs to both scans.}
\label{tab:hedgehog-convergence}
\begin{tabular}{ccccc}
\toprule
$R$ & $\Delta r$ & $g_{\ast,R}(0^-)$ &
$g_{\ast,R}(2^-)$ & $g_{\ast,R}(1^+)$\\
\midrule
\multicolumn{5}{l}{Radius scan at fixed $\Delta r=0.02$}\\
$20$ & $0.02$ & $1.000303$ & $-2.000605$ & $1.481396$\\
$40$ & $0.02$ & $0.999975$ & $-1.999950$ & $1.481393$\\
$80$ & $0.02$ & $0.999934$ & $-1.999868$ & $1.481393$\\
\addlinespace
\multicolumn{5}{l}{Mesh scan at fixed $R=40$; the $\Delta r=0.02$ row is shared}\\
$40$ & $0.04$ & $0.999760$ & $-1.999520$ & $1.481126$\\
$40$ & $0.01$ & $1.000029$ & $-2.000058$ & $1.481459$\\
\midrule
continuum & - & $1$ & $-2$ & $40/27$\\
\bottomrule
\end{tabular}
\end{table}

The radius scan shows that the singlet and quintet are more sensitive to tail
truncation than the triplet at this mesh spacing. At $R=40$, the triplet errors
for $\Delta r=0.04$, $0.02$, and $0.01$ are approximately
$3.55\times10^{-4}$, $8.88\times10^{-5}$, and $2.21\times10^{-5}$; each halving
of the grid spacing reduces the absolute error by a factor close to four.
Successive mesh differences show the same second-order behavior in all three
channels. The singlet and quintet errors need not decrease monotonically because
finite-box and finite-difference contributions partially cancel. These scans
are therefore consistent with convergence to the analytic thresholds and with
the stated discretization order; solver tolerances remain implementation checks
rather than separate physical evidence.

\subsection{Continuum crossing quotients and boxed comparisons}
\label{subsec:hedgehog-crossing}

At a Dirichlet-box crossing, the Feynman-Hellmann derivative with respect to
$g$ is
\begin{equation}
  \Gamma_{JL,R}
  =c_{JL}
  \frac{\displaystyle\int_0^R\phi(r)u_{JL,R}(r)^2\,\dd r}
       {\displaystyle\int_0^R u_{JL,R}(r)^2\,\dd r}.
  \label{eq:boxed-crossing-quotient}
\end{equation}
For the threshold states in Eq.~\eqref{eq:henyey-zero-and-wall}, both continuum
integrals converge, and the corresponding Rayleigh quotient is
\begin{equation}
  \Gamma_{JL}^{(\infty)}
  =c_{JL}
  \frac{\displaystyle\int_0^\infty\phi(r)u_L(r)^2\,\dd r}
       {\displaystyle\int_0^\infty u_L(r)^2\,\dd r}.
  \label{eq:continuum-crossing-quotient}
\end{equation}
At the continuum threshold, this quotient is not the derivative of an isolated
analytic eigenvalue because the zero mode lies at the lower edge of the
essential spectrum. Standard Dirichlet Sturm-Liouville convergence on
expanding intervals identifies the nodeless box states with $u_L$ on compact
subintervals~\cite{Zettl2005}; convergence of the normalized quotients over the
expanding domains additionally requires uniform $L^2$ control of their tails.
Equation~\eqref{eq:continuum-crossing-quotient} is therefore used as the
analytic continuum comparison for Eq.~\eqref{eq:boxed-crossing-quotient}.

The change of variables $t=r^3$ reduces the required integrals to Euler beta
functions. For the $0^-$ singlet,
\begin{equation}
  \Gamma_{0,1}^{(\infty)}
  =-\frac{9\sqrt3}{4\pi}
  =-1.240490\ldots,
  \label{eq:singlet-gamma}
\end{equation}
whereas the $(J,L)=(1,2)$ triplet gives
\begin{equation}
  \Gamma_{1,2}^{(\infty)}
  =-63B\left(\frac83,\frac83\right)
  =-3.553945\ldots.
  \label{eq:triplet-gamma}
\end{equation}
For the $(J,L)=(2,1)$ quintet,
\begin{equation}
  \Gamma_{2,1}^{(\infty)}
  =+\frac{9\sqrt3}{8\pi}
  =+0.620245\ldots.
  \label{eq:quintet-gamma}
\end{equation}
The last sign refers to increasing $g$. When the outward negative ray is
parametrized by $t=-g$, its velocity is
$-\Gamma_{2,1}^{(\infty)}<0$, consistent with the affine-ray bound in every
finite box. Appendix~\ref{app:beta} gives the integral evaluation.

\subsection{Protected multiplets and a normalized representation-space illustration}
\label{subsec:triplet}

The threshold at $g=1$ is a $J=0$ singlet, the threshold at $g=-2$ is a $J=2$
quintet, and the next positive threshold, at $g=40/27$, is the $(J,L)=(1,2)$
triplet. In a rotationally invariant box, the quintet and triplet carry real
irreducible representations of the diagonal $SO(3)$. Their crossing matrices
are therefore scalar by Theorem~\ref{thm:symmetry-protected}; the scalar
continuum limits are given by Eqs.~\eqref{eq:quintet-gamma} and
\eqref{eq:triplet-gamma}.

For the triplet, the full angular product space is
$L=2\otimes S=1$ and has dimension fifteen. Exact angular-momentum addition
gives
\begin{equation}
  \spec(\bm S\!\cdot\!\bm L)
  =\{-3\ (J=1,\ 3\text{-fold}),\
     -1\ (J=2,\ 5\text{-fold}),\
     +2\ (J=3,\ 7\text{-fold})\}.
  \label{eq:SL-spectrum}
\end{equation}
Let $P_1$ denote the projector onto the $J=1$ subspace, with basis ordered by
$m=(-1,0,+1)$. Projection of two exact angular operators gives
\begin{align}
  D_{\mathrm{even}}
  &=P_1(3S_3L_3-\bm S\!\cdot\!\bm L)P_1
  \notag\\
  &=\operatorname{diag}\left(-\frac35,\frac65,-\frac35\right)
   =-\frac35\operatorname{diag}(1,-2,1),
  \label{eq:even-triplet-splitting}\\
  D_{\mathrm{odd}}
  &=P_1J_3P_1
   =\operatorname{diag}(-1,0,+1).
  \label{eq:odd-triplet-splitting}
\end{align}
The first matrix is invariant under $m\leftrightarrow-m$ and displays the
$3\to2\oplus1$ decomposition under axial $O(2)$. The second changes sign under
that reflection and separates the two components of the remaining doublet.

These matrices define a representation-space parametrization of the allowed
first-order diagonal patterns,
\begin{equation}
  \Gamma_{\mathrm{eff}}(\epsilon,\eta)
  =\gamma\one_3
   +\kappa_{\mathrm e}\epsilon D_{\mathrm{even}}
   +\kappa_{\mathrm o}\eta D_{\mathrm{odd}}
   +O\!\left(\lVert(\epsilon,\eta)\rVert^2\right),
  \label{eq:effective-triplet-form}
\end{equation}
with $\gamma$ fixed by Eq.~\eqref{eq:triplet-gamma}. Angular momentum algebra
determines the matrix patterns but does not determine the radial coefficients
$\kappa_{\mathrm e}$ and $\kappa_{\mathrm o}$. Their calculation would require
an explicit transverse deformation $B_i^a(x)$ and the projection of
$-\partial_i\ad_{B_i}$ onto the full radial-angular threshold states. No such
deformation is constructed here.

Figure~\ref{fig:triplet} adopts $\kappa_{\mathrm e}=\kappa_{\mathrm o}=1$ to
display the two angular components; $\epsilon$ and $\eta$ are dimensionless
plotting parameters.

\begin{figure}[htbp]
  \centering
  \includegraphics[width=\linewidth]{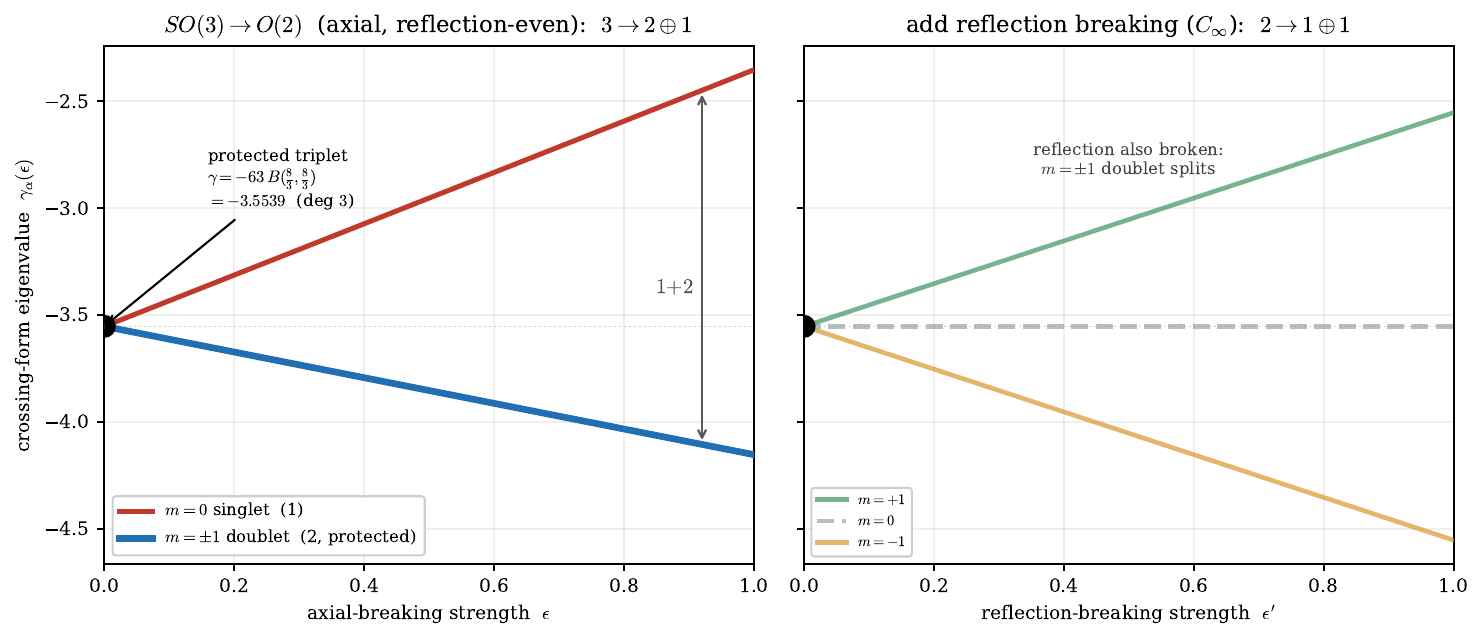}
  \caption{Normalized representation-space illustration for the $J=1$, $L=2$
  triplet. The common intercept is the continuum quotient
  $\gamma=-63B(8/3,8/3)$. The reflection-even matrix
  $D_{\mathrm{even}}$ produces an $m=0$ singlet and an $m=\pm1$ doublet; the
  reflection-odd matrix $D_{\mathrm{odd}}$, plotted separately about the same
  intercept, separates the doublet. The convention
  $\kappa_{\mathrm e}=\kappa_{\mathrm o}=1$ fixes only the graphical
  normalization. The figure is not a diagonalization of a radially deformed
  Faddeev-Popov operator.}
  \label{fig:triplet}
\end{figure}

\section{Discussion}
\label{sec:discussion}

The reduced formulation distinguishes the fixed global-color kernel from the
additional kernel that defines the nontrivial horizon. Constant ghosts make the
critical set of the orbit-norm functional Morse-Bott, so its normal Hessian acts
on $\Hilb_0$. The focal index, however, must be computed from the full pencil
along the segment; reduction is performed only at the endpoint through
Eq.~\eqref{eq:focal-endpoint-splitting}. This ordering prevents an intermediate
compression from introducing tangencies that do not belong to the geometric
pencil.

The distinction between $\Delta_A$ and $\Mfp[A]$ then acquires a direct
geometric interpretation. The inverse covariant Laplacian extracts the
orthogonal orbit component and enters the curvature of the orthogonal
connection. The Faddeev-Popov operator is instead the Hessian associated with
the fixed Landau slice and becomes singular when the orbit is tangent to that
slice. Equation~\eqref{eq:two-operator-identity-b} relates the operators but
does not identify their kernels: the lower-order term records the contribution
of the second fundamental form after pullback to ghost parameters. This
separation prevents a local ghost pole from being mistaken for a failure of
the orthogonal connection or for Singer's global topological obstruction.

At an isolated crossing, the critical projector and the crossing form provide
the finite-dimensional data needed to describe the local singularity. For a
simple zero, $B\mapsto\Gamma_A(B;f_0)$ is the conormal to the spectral wall;
for a degenerate zero, the eigenvalues of $P\delta_B\Mfp P$ are the velocities
of the analytic branches. The same pair $(P,\Gamma)$ determines the leading
resolvent coefficient $P\Gamma^{-1}P$ whenever the crossing form is invertible.
The source projection remains part of the statement: an operator pole may be
present even when a chosen source has no component in the critical space. The
rank-one and matrix Feshbach reductions explain why coupling to the
noncritical spectrum does not renormalize the first-order coefficient at
threshold, since the off-diagonal coupling vanishes there.

When $\Gamma$ has a kernel, its Moore-Penrose pseudoinverse accounts only for
the regular, linearly crossing part of the critical space, so the analysis of
a tangential path requires additional information. Eliminating the complementary spectrum
must be followed by a second Schur reduction that removes this regular critical
subspace. The first invertible coefficient of the resulting operator
$F_0(\delta)$ fixes the pole order in the tangential sector. The finite-matrix
example in Eq.~\eqref{eq:second-schur-example} shows that off-diagonal coupling
can lower the order inferred from the direct tangential compression. Pole
orders are therefore properties of the twice-reduced analytic family, not of
$\Ker\Gamma$ alone.

Affinity supplies a separate constraint on crossings reached from the positive
region. Concavity of the lowest eigenvalue recovers convexity of $\Omega$, while
Eq.~\eqref{eq:matrix-affine-bound} acts on the entire critical block and forces
every radial branch velocity to be negative. This operator identity includes
symmetry-protected multiplets and requires no simplicity assumption. It does
not prohibit a different path through the same configuration from being
tangent to the wall, so the affine-ray theorem and the tangential reduction
address complementary path geometries.

The hedgehog family provides the continuum application of the operator
framework. The
closed expression in Eq.~\eqref{eq:henyey-zero-and-wall} gives one normalizable,
nodeless threshold state in every coupled spin-orbit channel of a single
background family. Nodelessness identifies the first threshold of each radial
operator, and the explicit $J,L$ dependence permits minimization over the
complete channel tower. The interval $-2<g<1$ is consequently an analytic
all-channel statement; it is not inferred from the selected curves in
Fig.~\ref{fig:hedgehog-walls}. Its endpoints have different representation
content: the positive ray reaches a $J=0$ singlet, whereas the negative ray
reaches a symmetry-protected $J=2$ quintet.

The compact and continuum results concern different spectral regimes. On
$\mathbb R^3$, the hedgehog zero modes lie at the lower edge of the essential
spectrum, and the isolated-eigenvalue Laurent expansion used on a compact
manifold is not available without a limiting prescription. Dirichlet boxes
restore a discrete Sturm-Liouville spectrum and allow the crossing forms to be
computed before taking $R\to\infty$. Their roots approach the exact thresholds,
but the radial box alone does not specify gauge-compatible boundary conditions
for the complete configuration space. It therefore connects the spectral
calculations across the finite- and infinite-volume settings without claiming
to reproduce the full finite-domain Singer geometry.

The two auxiliary examples test different components of the framework. The
charged Mathieu block checks the finite-volume crossing identity and projected
spectral convergence in a reducible sector, whereas the triplet matrices give
the exact angular patterns allowed by the residual symmetries. Determining the
physical triplet slopes requires an explicit transverse deformation and matrix
elements evaluated on the full radial-angular threshold states.

The finite-volume formulation suggests a direct comparison with numerical
gauge-field studies. Given an interpolation $A_t$ through a discretized
horizon, one can compute the low-mode projector, the compressed derivative
$P\dot{\Mfp}P$, and the response of $\Mfp^{-1}J$ for sources with controlled
critical overlap. Equation~\eqref{eq:matrix-affine-bound} then supplies a sign
and magnitude test for any affine interpolation that starts inside $\Omega$.
The free starting gap scales as $L^{-2}$, so the bound weakens with increasing
volume and does not imply a volume-independent crossing slope.

Several extensions require additional structure beyond the assumptions used
here. Gauge-compatible boundary conditions would connect the radial box to a
full finite-domain orbit geometry. An explicit axial deformation of the
hedgehog would determine $\kappa_{\mathrm e}$ and $\kappa_{\mathrm o}$ in
Eq.~\eqref{eq:effective-triplet-form}. Reducible strata require a joint treatment
of gauge stabilizers and the residual global-color critical manifold.
Background and linear covariant gauges may also lose the self-adjoint structure
on which the crossing-form and min-max arguments rely, although
BRST-invariant and background-independent horizon constructions provide
possible settings for a modified analysis
~\cite{CapriUniversal2018,JustoPereiraSobreiro2022}. The present work is
Euclidean and elliptic; the Lorentzian no-pole proposal of
~\cite{GuimaraesLorentzian2026} replaces the spectral boundary by a
real-time boundary-value problem and therefore lies outside its operator
framework.

A full lattice realization of orbit-space geometry would require more
than the radial Dirichlet discretization used here. Explicit coordinates,
the metric, the inverse metric, and the Laplace-Beltrami operator have
been constructed for finite rectangular $SU(2)$ lattices with open
boundary conditions in Ref.~\cite{LauferOrland2013}.

\section{Conclusion}
\label{sec:conclusion}

On the reduced ghost space of a compact Landau-gauge problem, a regular isolated
Gribov-horizon crossing is governed by a critical projector and an invertible
symmetric crossing form. These data orient the spectral-flow event, determine
the normal Morse-index jump, and fix the leading singular part of the sourced
ghost resolvent; for a one-dimensional critical space, they also specify the
conormal to a simple wall. The covariant Laplacian defining
the orthogonal connection remains distinct from the Faddeev-Popov Hessian,
while the full affine focal pencil relates the negative index of the latter to
additional focal directions after the residual global-color multiplicity is
separated at the endpoint.

Affine rays from the positive region have negative-definite crossing matrices,
including at symmetry-protected multiplets. General analytic paths may instead
be tangent to the wall; in that case, the pole order is determined only after
the regular critical directions have been eliminated by a second Schur
complement. This distinction separates the universal local operator statements
from path-dependent higher-order contact.

For the $SU(2)$ hedgehog family, a closed-form threshold state exists in every
coupled spin-orbit channel. Minimization over all channels yields the exact
threshold-stability interval $(-2,1)$, bounded by singlet and quintet zero
modes. Dirichlet boxes approach these continuum values and provide a controlled
finite-volume comparison, while the projected Mathieu and triplet examples
illustrate, respectively, selection-rule degeneracy and symmetry-allowed
splitting patterns. The scope remains local: the quantum horizon measure and
global gauge fixing require additional input. Within that scope, the
projectors, crossing matrices, and finite-volume limits are directly testable
near a discretized Gribov boundary.

\appendix

\section{Beta-function evaluation of the hedgehog crossing quotients}
\label{app:beta}

The continuum quotients in Eq.~\eqref{eq:continuum-crossing-quotient} reduce to
Euler beta functions through the substitution $t=r^3$. We use
\begin{equation}
  \int_0^\infty \frac{t^{x-1}}{(1+t)^{x+y}}\,\dd t
  =B(x,y).
  \label{eq:beta-definition}
\end{equation}
For $L=1$, the normalization integral is
\begin{equation}
  \int_0^\infty u_1(r)^2\,\dd r
  =\int_0^\infty\frac{r^4}{(r^3+1)^2}\,\dd r
  =\frac13B\left(\frac53,\frac13\right)
  =\frac{4\pi}{9\sqrt3}.
  \label{eq:u1-denominator}
\end{equation}
The corresponding numerator is
\begin{equation}
  \int_0^\infty\phi(r)u_1(r)^2\,\dd r
  =3B(2,2)=\frac12.
  \label{eq:u1-numerator}
\end{equation}
Multiplication of their ratio by $c_{01}=-2$ gives
Eq.~\eqref{eq:singlet-gamma}; using $c_{21}=+1$ instead gives
Eq.~\eqref{eq:quintet-gamma}.

For $L=2$, the two integrals are
\begin{align}
  \int_0^\infty u_2(r)^2\,\dd r
  &=\frac13B\left(\frac73,1\right)=\frac17,
  \label{eq:u2-denominator}\\
  \int_0^\infty\phi(r)u_2(r)^2\,\dd r
  &=3B\left(\frac83,\frac83\right).
  \label{eq:u2-numerator}
\end{align}
Since $c_{12}=-3$, the ratio is
$-63B(8/3,8/3)$, as stated in Eq.~\eqref{eq:triplet-gamma}.

\end{document}